\documentclass[1p,final]{elsarticle}
\usepackage{amssymb,color,pslatex}
\usepackage{amsthm}
 \usepackage{stmaryrd}
 \usepackage{amsfonts}
  \usepackage{mathrsfs}
  \usepackage{latexsym,bm}
   \usepackage{yfonts}
   \usepackage[all,poly,knot]{xy}

 \newtheorem{theorem}{Theorem}[section]
\newtheorem{Legendre theorem}{Legendre Theorem}[section]
\newtheorem{lemma}[theorem]{Lemma}
\newtheorem{corollary}[theorem]{Corollary}

\newcommand{\Tr}{{\rm Tr}}
\newcommand{\gf}{ {{\mathbb F}} }

\begin{document}

\begin{frontmatter}



\title{Further Results on Permutation Polynomials over Finite Fields}
\tnotetext[fn1]{P. Yuan's research was supported by the NSF of China (Grant No. 11271142) and
the Guangdong Provincial Natural Science Foundation(Grant No. S2012010009942). C. Ding's research was supported by
The Hong Kong Research Grants Council, Proj. No. 601013.}


\author[ypz]{Pingzhi Yuan}
\ead{mcsypz@mail.sysu.edu.cn}
\author[cding]{Cunsheng Ding}
\ead{cding@ust.hk}

\address[ypz]{School of Mathematics, South China Normal University, Guangzhou 510631, China}
\address[cding]{Department of Computer Science
                                                  and Engineering, The Hong Kong University of Science and Technology,
                                                  Clear Water Bay, Kowloon, Hong Kong, China}

\begin{abstract}
Permutation polynomials are an interesting subject of mathematics and have applications in other areas of
mathematics and engineering.
In this paper,  we develop general theorems on permutation polynomials over finite fields. As a demonstration
of the theorems, we present a number of classes of explicit permutation polynomials on $\gf_q$.

\end{abstract}

\begin{keyword}
Cyclic codes \sep polynomials \sep permutation polynomials \sep skew Hadamard difference sets.

\MSC  11C08 \sep 12E10

\end{keyword}

\end{frontmatter}

\section{Introduction}

Let  $\gf_q$ be the finite field with $q$ elements,  where $q$ is a prime power, and
let $\gf_q[x]$
be the ring of polynomials in a single indeterminate $x$ over $\gf_q$. A polynomial
$f \in\gf_q[x]$ is called a {\em permutation polynomial} (PP) of $\gf_q$ if it induces
a one-to-one map from $\gf_q$ to itself. A permutation on $\gf_q$ is a bijection 
from $\gf_q$ to itself. It is well known that every permutation on $\gf_q$ can be 
expressed as a permutation polynomial over $\gf_q$. 

Permutation polynomials over finite fields have been a hot topic of study for many  
years, and have applications in coding theory \cite{DH13,LC07}, cryptography 
\cite{LM84,SH,RSA}, combinatorial designs \cite{DY06}, and other areas of mathematics 
and engineering. For example, Dickson permutation polynomials of order five, i.e.,  
$D_5(x, a)=x^5+ax^3-a^2x$ over finite fields,  led to a 70-year research breakthrough in 
combinatorics \cite{DY06}, gave a family of perfect nonlinear functions for cryptography 
\cite{DY06}, generated good linear codes \cite{CDY,YCD} for data communication and 
storage, and produced optimal signal sets for CDMA communications \cite{DingYin}, to 
mention only a few applications of these Dickson permutation polynomials.     
Information on constructions,
properties and applications of permutation polynomials may be
found in Lidl and Niederreiter \cite{LR97}, and Mullen \cite{Mull}.

The trace function $\Tr(x)$ from $\mathbb{F}_{q^n}$ to $\mathbb{F}_{q}$
is defined by
$\Tr(x)=x+x^q+x^{q^{2}}+\cdots+x^{q^{n-1}}.$
A number of classes of permutation polynomials related to the trace
functions were constructed in \cite{CK08, CK09,Ky10,Ma10,Zi10}.

Recently, Akbary, Ghioca and Wang derived a lemma about permutations on finite sets \cite{AGW}, 
which contains Lemma 2.1 in \cite{Zi08} and Proposition 3 in \cite{Zi10} as special cases, and 
employed this lemma to unify some earlier constructions and developed new constructions of
permutation polynomials over finite fields.
In \cite{YD}, with this lemma we derived several theorems about permutation
polynomials over finite fields. These
theorems give not only a further unified treatment of some of the earlier
constructions of permutation polynomials, but also new specific permutation
polynomials.

In this paper,  we continue our investigations in \cite{YD} by employing this lemma in \cite{AGW} again.
We  first develop generic theorems on permutation polynomials over finite fields with this powerful lemma.
We then construct new permutation polynomials of explicit forms.

\section{Auxiliary results \& the main Lemma}

In this section, we present some auxiliary results that will be needed in the
sequel. Throughout this paper $p$ is a prime and $q=p^e$ for a positive integer $e$.

A polynomial of the form
$$
L(x)= \sum_{i=0}^{n-1} a_{i}x^{q^{i}} \in\mathbb{F}_{q^n}[x]
$$
is called a {\em $q$-polynomial} over $\gf_{q^n}$, and is a
permutation polynomial on $\mathbb{F}_{q^n}$ if and only if the
circulant matrix
\begin{eqnarray}\label{eqn-matP}
A=\left(\begin{array}{lllll}
               a_0 &a_1 &a_2&\cdots &a_{n-1}\\
               a_{n-1}^q &a_0^q   &a_1^q  &\cdots &a_{n-2}^q\\
               a_{n-2}^{q^2} &a_{n-1}^{q^2}   &a_{0}^{q^2} &\cdots & a_{n-3}^{q^2}\\
               \cdots\\
               a_{1}^{q^{n-1}} &a_{2}^{q^{n-1}}&a_3^{q^{n-1}}&\cdots&a_0^{q^{n-1}}\end{array}\right)
\end{eqnarray}
has nonzero determinant (see [9, p.362]). In most cases it is not convenient to
use this result to find out permutation $q$-polynomials, as it may be hard
to determine if the determinant of this matrix is nonzero \cite{DXY}. Hence
it would be interesting to develop other approaches to the construction of
permutation $q$-polynomials.

In the sequel we need the following Lemma whose proof is straightforward.

\begin{lemma}
Let $L(x)= \sum_{i=0}^{n-1} a_{i}x^{q^{i}}  \in\mathbb{F}_{q}[x]$
be a $q$-polynomial and let $\Tr(x)$ be the trace function from  $\mathbb{F}_{q^n}$
to $\mathbb{F}_{q}$.
Then, for each $\alpha\in \mathbb{F}_{q^n}$, we have
$$L(\Tr(\alpha))=\Tr(L(\alpha))=\left(\sum_{i=0}^{n-1} a_{i}\right)\Tr(\alpha).$$
\end{lemma}

The polynomials
$$l(x)=\sum_{i=0}^ma_ix^i \qquad {\mbox and}\qquad
L(x)=\sum_{i=0}^ma_ix^{q^i}$$ over  $\mathbb{F}_{q^n}$ are called the
$q$-associate of each other. More specifically, $l(x)$ is the
\emph{conventional} $q$-associate of $L(x)$ and $L(x)$ is the linearized
$q$-associate of $l(x)$ \cite[p. 115]{LR86}.

The following lemma is also needed in the sequel.

\begin{lemma}\label{lem-gcdaa} {\rm (\cite[p. 109]{LR86})}
Let $L_1(x)$ and $L_2(x)$ be two $q$-polynomials over $\mathbb{F}_q$,
and let $l_1(x)$ and $l_2(x)$ be the $q$-associate polynomials over
$\mathbb{F}_q$. Then the common roots of $L_1(x)=0$ and $L_2(x)=0$
are all the roots of the linearized $q$-associate of $\gcd(l_1(x),
l_2(x))$. In particular, $x=0$ is the only common root of $L_1(x)=0$
and $L_2(x)=0$ in any finite extension of  $\mathbb{F}_q$ if and
only if $\gcd(l_1(x), l_2(x))=1$.
\end{lemma}

The following lemma was developed by Akbary, Ghioca, and Wang  \cite[Lemma 1.1]{AGW},
and contains Lemma 2.1 in \cite{Zi08} and Proposition 3 in \cite{Zi10} as special cases. It will 
be frequently employed in the sequel.

\begin{lemma}\label{lem-1.1}
Let $A, S$ and $\bar{S}$ be finite sets with $\sharp
S=\sharp\bar{S}$, and let $f:A\rightarrow A$, $h:
S\rightarrow\bar{S}$, $\lambda: A\rightarrow S$, and $\bar{\lambda}:
A\rightarrow \bar{S}$ be maps such that $\bar{\lambda}\circ
f=h\circ\lambda$. If both $\lambda$ and $\bar{\lambda}$ are
surjective, then the following statements are equivalent:

(i) $f$ is bijective (a permutation of $A$); and

(ii) $h$ is bijective from $S$ to $\bar{S}$ and $f$ is
injective on $\lambda^{-1}(s)$ for each $s\in S$.
\end{lemma}

\section{Generic theorems on permutation polynomials}\label{sec-mainthms}

 The following lemma is an application of Lemma \ref{lem-1.1}, and
  is a variant of Theorem 1.4 (c) and Theorem 5.1 (c) in \cite{AGW}.

 \begin{lemma}\label{thm-main31}(\cite{YD} Theorem 6.1)
 Assume that $A$ is a  finite field and $S, \bar{S}$ are finite subsets  of $A$ with $\sharp(S)=\sharp(\bar{S})$ such that the
 maps
 $\psi: A\to S$ and $\bar{\psi}: A \to\bar{S}$ are surjective and $\bar{\psi}$ is additive, i.e.,
 $$\bar{\psi}(x+y)=\bar{\psi}(x)+\bar{\psi}(y) \mbox{ for all } x, y \in A.$$ Let $f: A \to A$ and $h: S \to\bar{S}$
 be maps such that the following diagram commutes:
$$
\xymatrix{
  A \ar[d]_{\psi} \ar[r]^{f}
                & A \ar[d]^{\bar{\psi}}  \\
  S \ar[r]_{h}
                & \bar{S}      }
$$
Then for any map $g: S\to A$, the map
$p(x)=f(x)+g(\psi(x))$ permutes $A$ if and only if

  i) $h$ is a bijection; and

  ii) $f$ is a injection on $\psi^{-1}(s)$ for every $s\in S$.

  Furthermore, if $ \bar{\psi}(g(\psi(x)))=0$ for every $x\in A$, then the map
$p(x)=f(x)+g(\psi(x))$ permutes $A$ if and only if
  $f$ permutes  $A$.

\end{lemma}

The following theorem is another application of Lemma \ref{lem-1.1}, and
is a variant of Lemma \ref{thm-main31}.

 \begin{theorem}\label{thm-main3}
 Assume that $A$ is a  finite field and $S, \bar{S}$ are finite subsets  of $A$ with $\sharp(S)=\sharp(\bar{S})$ such that the
 maps
 $\psi: A\to S$ and $\bar{\psi}: A \to\bar{S}$ are surjective and $\bar{\psi}$ is additive, i.e.,
 $$\bar{\psi}(x+y)=\bar{\psi}(x)+\bar{\psi}(y) \mbox{ for all } x, y \in A.$$ Let $u: A\to A$  and $v: A \to A$
 be maps such that the following diagram commutes:
$$
\xymatrix{
  A \ar[d]_{\psi} \ar[r]^{u + v}
                & A \ar[d]^{\bar{\psi}}  \\
  S \ar[r]_{h}
                & \bar{S}      }
$$
Assume also that $ \bar{\psi}(v(x))=0$ for every $x\in A$ and $v(x)$ is a constant on each $\psi^{-1}(s)$
 for all $s \in S$.
 Then the map
$f(x)=u(x)+v(x)$ permutes $A$ if and only if
  $u$ permutes  $A$.

\end{theorem}

 \begin{proof}
 It follows from Lemma \ref{lem-1.1} and the assumptions of this theorem
 that $u(x)+ v(x)$ permutes $A$ if and only if $h$ is a bijection from
 $S$ to $\bar{S}$ and $u(x)+ v(x)$ is injective on each $\psi^{-1}(s)$
 for all $s \in S$.

On the other hand, by assumption we have $ \bar{\psi}(v(x))=0$ for every $x\in A$.
Hence,
$$
\bar{\psi}(u(x) + v(x))=\bar{\psi}(u(x)) + \bar{\psi}(v(x))=\bar{\psi}(u(x))
$$
for all $x \in A$.  Therefore, the following diagram commutes:
$$
\xymatrix{
  A \ar[d]_{\psi} \ar[r]^{u}
                & A \ar[d]^{\bar{\psi}}  \\
  S \ar[r]_{h}
                & \bar{S}      }
$$
Applying  Lemma \ref{lem-1.1} to this commutative diagram, we know that
$u(x)$ permutes $A$ if and only if $h$ is a bijection from $S$ to $\bar{S}$
and $u(x)$ is injective on each $\psi^{-1}(s)$  for all $s \in S$.

For each $s \in S$, $v(x)$ is a constant function on $\psi^{-1}(s)$ by assumption.
It then follows that  $u(x)+ v(x)$ is injective on each $\psi^{-1}(s)$
 for all $s \in S$ if and only if $u(x)$ is injective on each $\psi^{-1}(s)$
 for all $s \in S$.

Summarizing the discussions above proves the desired conclusion.
\end{proof}

As an application of Theorem \ref{thm-main3}, we have the following corollary.

\begin{corollary}\label{pro-main1}
Let $g(x)$ be a polynomial over $\gf_{q^n}$ such that $g(x)^q=g(x)$ for every $x\in\mathbb{F}_{q^n}$, and let $L(x)\in\mathbb{F}_q[x]$ be a linearized polynomial. Then for every $\delta\in\mathbb{F}_{q^n}$, the polynomial
$$f(x)=g(x^q-x+\delta)+L(x)$$
permutes $\mathbb{F}_{q^n}$ if and only if $L(x)$ permutes
$\mathbb{F}_{q^n}$.
\end{corollary}

\begin{proof}  We now consider Theorem \ref{thm-main3} and let $A=\gf_q$. We first define
\begin{eqnarray*}
S=\{x^q-x-\delta: x \in A\} \mbox{ and } \bar{S}=\{x^q-x: x \in A\}.
\end{eqnarray*}
It is easily seen that $\#(S)=\#(\bar{S})=q^{n-1}$.

We then define
\begin{eqnarray*}
\psi(x)=x^q-x-\delta, \mbox{ } \bar{\psi}(x)=x^q-x, \mbox{ and } h(x)=L(x)-L(\delta).
\end{eqnarray*}
By definition $\psi$ is a surjection from $A$ to $S$, $\bar{\psi}$ is a surjection from $A$ to $\bar{S}$
and is additive,
and $h$ is a function from $S$ to $\bar{S}$.

Define $u(x)=L(x)$ and $v(x)=g(x^q-x+\delta)$ for all $x \in A$.  Then $u(x)$ and $v(x)$ are
functions from $A$ to $A$. It is straightforward to verify that $\bar{\psi}(u(x)+v(x))=h(\psi(x))$
for all $x \in A$. Hence the diagram in Theorem \ref{thm-main3} commutes.

By definition and assumption we have that $ \bar{\psi}(v(x))=0$ for every $x\in A$ and
$v(x)$ is a constant on each $\psi^{-1}(s)$.

Hence all the conditions in Theorem \ref{thm-main3} are satisfied. Then the desired conclusion
of this corollary follows from Theorem \ref{thm-main3}.
\end{proof}

In order to apply Corollary \ref{pro-main1} for the construction of explicit permutation polynomials
on $\gf_{q^n}$, we need to find polynomials $g(x) \in \gf_{q^n}[x]$ such that $g^q=g$. We now search
for such polynomials $g(x)$.

Let $d$ be a divisor of $n$ and let $n=kd$. For any $d$ with $1<d<n$, define
$$M=M(n,d)=1+q^d+\cdots+q^{(k-1)d}.$$
Let $h(x)$ be any polynomial over $\gf_{q^n}$.
The following are examples of polynomials $g(x)$ such that $g(x)^q=g(x)$ for every $x\in\mathbb{F}_{q^n}$.

1. For $d=1$, let $g(x)=\Tr(h(x))$.

2. For $d=n$, let $g(x)=h(x)^{s(q^n-1)/(q-1)}$.

3. For $1<d<n$, let $g(x)=h(x)^{M}+h(x)^{Mq}\cdots+h(x)^{Mq^{d-1}}$.

4. If $g_1(x)$ and $g_2(x)$ are polynomials with $g_1(x)^q=g_1(x)$ and $g_2(x)^q=g_2(x)$, then we have
$$(g_1(x)g_2(x))^q=g_1(x)g_2(x) \mbox{ and } (g_1(x)+g_2(x))^q=g_1(x)+g_2(x).$$

5. In general, if $g_i(x),\ i=1, 2, \ldots, r,$ are polynomials over $\mathbb{F}_{q^n}[x]$ with
$g_i(x)^q=g_i(x),\  i=1, 2, \ldots, r$ and $g(x_1, \ldots, x_r)\in \mathbb{F}_{q}[x_1, \ldots, x_r]$,
then
$$ g(g_1(x), \ldots, g_r(x))^q=g(g_1(x), \ldots, g_r(x)).$$

Hence, there are many polynomials $g(x) \in \gf_{q^n}[x]$ such that $g^q=g$. In addition, there are a large
number of
linearized permutation polynomials $L(x) \in \gf_{q^n}$. Hence, Corollary \ref{pro-main1} leads to a lot of
new permutation polynomials over $\gf_{q^n}$ of the form $g(x^q-x+\delta)+L(x)$.

\vspace*{.2cm}
As an application of Theorem \ref{thm-main3}, we have the  following, which is different from  Theorem 4
in \cite{ZH}.

\begin{theorem}\label{thm-yding}
Let $t$ be  an even integer  and $n=2k$. Let $\delta\in\mathbb{F}_{q^n}$ with $\delta^{q^k}=-\delta$,
and let $L(x) \in \mathbb{F}_{q^k}[x]$. Then $f(x)=(x^{q^k}-x+\delta)^t
+L(x)$ is a PP over $\mathbb{F}_{q^n}$ if and only if $L(x)$  a PP over $\mathbb{F}_{q^n}$.
\end{theorem}

\begin{proof}  We now consider Theorem \ref{thm-main3} and let $A=\gf_q$. We first define
\begin{eqnarray*}
S= \bar{S}=\{x^{q^k}-x: x \in A\}.
\end{eqnarray*}

We then define
\begin{eqnarray*}
\psi(x)= \bar{\psi}(x)=x^{q^k}-x \mbox{ and } h(x)=L(x).
\end{eqnarray*}
By definition $\psi$ is a surjection from $A$ to $S$, $\bar{\psi}$ is a surjection from $A$ to $\bar{S}$
and is additive,
and $h$ is a function from $S$ to $\bar{S}$.

Define $u(x)=L(x)$ and $v(x)=(x^{q^k}-x+\delta)^t$ for all $x \in A$.  Then $u(x)$ and $v(x)$ are
functions from $A$ to $A$. It is straightforward to verify that $\bar{\psi}(u(x)+v(x))=h(\psi(x))$
for all $x \in A$. Hence the diagram in Theorem \ref{thm-main3} commutes.

By definition and assumption we have that $ \bar{\psi}(v(x))=0$ for every $x\in A$ and
$v(x)$ is a constant on each $\psi^{-1}(s)$.

Hence all the conditions in Theorem \ref{thm-main3} are satisfied. Then the desired conclusion
of this corollary follows from Theorem \ref{thm-main3}.
\end{proof}

\begin{corollary}\label{thm-ypz16}
Let $s, \, t, \, n,\, k$ be nonnegative integers such that  $2|t$ and $n=2k$. Let $\beta \in\mathbb{F}_{q^k}$,
$\gamma\in\mathbb{F}_{q}^*$,
and $\delta\in\mathbb{F}_{q^n}$ such that $\delta^{q^k}=-\delta$. Then $f(x)=(x^{q^k}-x+\delta)^t
+ \beta \Tr(x)  + \gamma x^{q^s}$ is a PP over $\mathbb{F}_{q^n}$  if and only if $\Tr(\beta \gamma^{-1})+1 \ne 0$.
\end{corollary}

\begin{proof}
Let $L(x)=\beta \Tr(x)  + \gamma x^{q^s}$. Then $L(x) \in \gf_{q^k}[x]$ is a linearized polynomial.
It then follows from Theorem \ref{thm-yding} that $f(x)$ is a PP over $\gf_{q^n}$ if and only if
$L(x)$ is a PP over $\gf_{q^n}$.

Since $L(x)$ is a linearized polynomial, $L(x)$ is a PP over $\gf_{q^n}$ if and only if $\gamma\ne0$ and $ x^{q^s}+\beta \gamma^{-1} \Tr(x)$ is a PP over $\mathbb{F}_{q^n}$.

We have obviously  the following commutative diagram:
 $$
\xymatrix{
  \mathbb{F}_{q^n} \ar[d]_{\Tr(x)} \ar[r]^{x^{q^s}+\beta\gamma^{-1} \Tr(x)}
                & \mathbb{F}_{q^n} \ar[d]^{\Tr(x)}  \\
  S \ar[r]_{(1+\Tr(\beta\gamma^{-1})) x}
                & S      }
$$
Applying  Lemma \ref{lem-1.1} to this commutative diagram, we know that $ x^{q^s}+\beta\gamma^{-1} Tr(x)$
 permutes $\mathbb{F}_{q^n}$ if and only if $\gamma\ne 0$ and  $\Tr(\beta\gamma^{-1})\ne-1$.
 The desired conclusion then follows. This completes the proof.

\end{proof}

We have also the following conclusion.

\begin{theorem}\label{thm-yd18}
Let $t$ and $k$ be integers  and $n=2k$. Let $\delta\in\mathbb{F}_{q^k}$, where $q$ is odd.
Let $\alpha \in \gf_{q^n}$ with
$\alpha^{q^k}=-\alpha$ and  $\beta \in \gf_{q^n}$ with
$\beta^{q^k}=-\beta$. Let  $L(x)\in\mathbb{F}_{q^k}[x]$ . Then $f(x)=\alpha (x^{q^k}+x+\delta)^t +
\beta \Tr(x) +L(x)$ is a PP over $\mathbb{F}_{q^n}$ if and only if $L(x) $ is  a PP over $\mathbb{F}_{q^n}$.
\end{theorem}

\begin{proof}  We now consider Theorem \ref{thm-main3} and let $A=\gf_q$. We first define
\begin{eqnarray*}
S= \bar{S}=\{x^{q^k}+x: x \in A\}.
\end{eqnarray*}

We then define
\begin{eqnarray*}
\psi(x)= \bar{\psi}(x)=x^{q^k}+x \mbox{ and } h(x)=L(x).
\end{eqnarray*}
By definition $\psi$ is a surjection from $A$ to $S$, $\bar{\psi}$ is a surjection from $A$ to $\bar{S}$
and is additive,
and $h$ is a function from $S$ to $\bar{S}$.

Define $u(x)=L(x)$ and $v(x)=\alpha (x^{q^k}+x+\delta)^t + \beta \Tr(x)$ for all $x \in A$.  Then $u(x)$ and $v(x)$ are
functions from $A$ to $A$. It is straightforward to verify that $\bar{\psi}(u(x)+v(x))=h(\psi(x))$
for all $x \in A$. Hence the diagram in Theorem \ref{thm-main3} commutes.

By definition and assumption we have that $ \bar{\psi}(v(x))=0$ for every $x\in A$. Note that $q$ is odd.
We have $v(x)=\alpha (x^{q^k}+x+\delta)^t +2^{-1}  \beta \Tr(x+x^{q^k})$. Hence,
$v(x)$ is a constant on each $\psi^{-1}(s)$.

Hence all the conditions in Theorem \ref{thm-main3} are satisfied. Then the desired conclusion
of this corollary follows from Theorem \ref{thm-main3}.
\end{proof}

The following follows directly from Theorem \ref{thm-yd18}.

\begin{corollary}\label{cor-yd18}
Let $s$, $t$ and $k$ be integers  and $n=2k$. Let $\delta\in\mathbb{F}_{q^k}$ and $\gamma \in\mathbb{F}_{q^k}$, where $q$ is odd.
Let $\alpha \in \gf_{q^n}$ with
$\alpha^{q^k}=-\alpha$ and  $\beta \in \gf_{q^n}$ with
$\beta^{q^k}=-\beta$.  Then $f(x)=\alpha (x^{q^k}+x+\delta)^t +
\beta \Tr(x) + \gamma x^{q^s}$ is a PP over $\mathbb{F}_{q^n}$ if and only if $\gamma \neq 0$.
\end{corollary}

\begin{corollary}\label{pro-main2}
Let $q$ be odd.
Let $g(x)$ be a polynomial over $\gf_{q^n}$ such that $g(x)^q=-g(x)$ for every $x\in\mathbb{F}_{q^n}$, and
let $L(x)\in\mathbb{F}_q[x]$ be a linearized polynomial. Let $\beta \in \gf_{q^n}$ with $\beta^q=-\beta$. Then
for every $\delta\in\mathbb{F}_{q^n}$, the polynomial
$$f(x)=g(x^q+x+\delta)+\beta \Tr(x) + L(x)$$ permutes $\mathbb{F}_{q^n}$ if and only if $L(x)$ permutes $\gf_{q^n}$.
\end{corollary}

\begin{proof}  We now consider Theorem \ref{thm-main3} and let $A=\gf_q$. We first define
\begin{eqnarray*}
S=\{x^q+x+\delta: x \in A\} \mbox{ and } \bar{S}=\{x^q+x: x \in A\}.
\end{eqnarray*}
It is easily seen that $\#(S)=\#(\bar{S})$.

We then define
\begin{eqnarray*}
\psi(x)=x^q+x+\delta, \mbox{ } \bar{\psi}(x)=x^q+x, \mbox{ and } h(x)=L(x)-L(\delta).
\end{eqnarray*}
By definition $\psi$ is a surjection from $A$ to $S$, $\bar{\psi}$ is a surjection from $A$ to $\bar{S}$
and is additive,
and $h$ is a function from $S$ to $\bar{S}$.

Define $u(x)=L(x)$ and $v(x)=g(x^q+x+\delta) + \beta \Tr(x)$ for all $x \in A$.  Then $u(x)$ and $v(x)$ are
functions from $A$ to $A$. It is straightforward to verify that $\bar{\psi}(u(x)+v(x))=h(\psi(x))$
for all $x \in A$. Hence the diagram in Theorem \ref{thm-main3} commutes.

By definition and assumption we have that $ \bar{\psi}(v(x))=0$ for every $x\in A$ and
$$v(x)=g(x^{q}+x+\delta)+2^{-1}  \beta \Tr(x+x^{q})$$ is a constant on each $\psi^{-1}(s)$.

Hence all the conditions in Theorem \ref{thm-main3} are satisfied. Then the desired conclusion
of this corollary follows from Theorem \ref{thm-main3}.
\end{proof}

To apply Corollary \ref{pro-main2}, we have to find $\beta \in \gf_{q^n}$ such that $\beta^q=-\beta$ and
$g(x) \in \gf_{q^n}[x]$ with $g^q=-g$. Note that $q$ is odd. It can be proven that  $x^q=-x$ has only one
solution $x=0$ when $n$ is odd, and $q$ solutions when $n$ is even.

We now turn to the search for $g(x) \in \gf_{q^n}[x]$ with $g^q=-g$. Below are examples of such
polynomials $g(x)$.

1. Let $n=2k$ and $q$ be an odd prime power. For any $h(x) \in \gf_{q^n}[x]$, define
$$g(x)=h(x)^{q^{2k-1}}+h(x)^{q^{2k-3}}+\cdots+h(x)^{q}-h(x)^{q^{2k-2}}-h(x)^{q^{2k-4}}-\cdots-h(x).$$

2. Let $n$ be even and let $0\ne a$ be an element of $\mathbb{F}_{q^n}$ such that $a^q+a=0$. For any
$h(x) \in \gf_{q^n}[x]$ with $h(x)^q=h(x)$, the polynomial $g(x)=ah(x)$ satisfies that $g(x)^q=-g(x)$.

3. Let $k \ge 1$ and $d \ge 1$ be integers. Define $n=2kd$. Let
$$M=M(n,d)=1+q^{2d}+\cdots+q^{2(k-1)d}.$$
Then for any  polynomial $h(x)\in\mathbb{F}_{q^n}[x]$, the polynomial
$$g(x)=h(x)^{M}+h(x)^{Mq^2}\cdots+h(x)^{Mq^{2(d-1)}}-h(x)^{Mq}-h(x)^{Mq^3}\cdots-h(x)^{Mq^{2d-1}}$$
satisfies that
$g^q=-g$. For example, we have the following polynomials such that $g(x)^q=-g(x)$ for every $x\in\mathbb{F}_{q^n}$.

(3.i). For $d=1$, define
$$g(x)=h(x)^{1+q^2+\cdots+q^{2kd-2}}-h(x)^{q+q^3+\cdots+q^{2kd-1}},$$
where $h(x)\in\mathbb{F}_{q^n}[x]$.

(3.ii). For $d=n/2=k$, define
$$g(x)=h(x)^{q^{2k-1}}+h(x)^{q^{2k-3}}+\cdots+h(x)^{q}-h(x)^{q^{2k-2}}-h(x)^{q^{2k-4}}-\cdots-h(x).$$

(3.iii). For $1<d<k$, define
$$g(x)=h(x)^{M}+h(x)^{Mq^2}\cdots+h(x)^{Mq^{2(d-1)}}-h(x)^{Mq}-h(x)^{Mq^3}\cdots-h(x)^{Mq^{2d-1}}.$$

(3.iv). If $g(x)$ and $h(x)$ are polynomials with $g(x)^q=g(x)$ and $h(x)^q=-h(x)$, then we have
$$(g(x)h(x))^q=-g(x)h(x).$$

(3.v). If $g(x)$ and $h(x)$ are polynomials with $g(x)^q=-g(x)$ and $h(x)^q=-h(x)$, then we have
$$(g(x)+h(x))^q=-(g(x)+h(x)).$$

The discussions above show that Corollary \ref{pro-main2} yields many explicit permutation polynomials
of the form $g(x^q+x+\delta)+\beta \Tr(x) + L(x)$, where $L(x)$ is a linearized PP.

\begin{theorem}\label{thm-dy182}
Let $n=4k$  and $\delta$ be an element of $\mathbb{F}_{q^n}$.
 Let
 $g(x)=  \sum_{i=1}^k x^{q^{2(i-1)} + q^{2(i-1)+2k}}.$
 Then the polynomial
 $f(x)=g(x^q-x+\delta)+ax,$
 where $0\ne a\in\mathbb{F}_{q}$, permutes $\mathbb{F}_{q^n}$ if and only if  $\Tr(\delta)\ne a$
\end{theorem}

\begin{proof} It follows from Lemma \ref{lem-1.1} and the following commutative diagram
$$
\xymatrix{
  \mathbb{F}_{q^n} \ar[d]_{x^q-x+\delta} \ar[r]^{f}
                & \mathbb{F}_{q^n} \ar[d]^{x^q-x}  \\
  S \ar[r]_{g(x)^q-g(x)+ax-a\delta}
                & \bar{S}      }
$$
where $S=\{b^q-b+\delta: b\in\mathbb{F}_{q^n}\}$ and $\bar{S}=\{b^q-b: b\in\mathbb{F}_{q^n}\} =a\bar{S}$,
that the polynomial
 $$f(x)=g(x^q-x+\delta)+ax,$$
  permutes $\mathbb{F}_{q^n}$ if and only if $g(x)^q-g(x)+ax-a\delta$ is a bijection from $S$ to $\bar{S}$.

Let $ab, b\in \bar{S}$ be any element in $\bar{S}$. We want to show that the equation
  \begin{equation}\label{eq32}
  g(x)^q-g(x)+ax-a\delta=ab
  \end{equation}
  has at most one solution. Note that $g(x)^{q^2}=g(x)$. Raising both sides of (\ref{eq32}) to the power of $q$, 
  we obtain
   \begin{equation}\label{eq322}
  g(x)-g(x)^q+ax^q-a\delta^q=ab^q.
  \end{equation}
  Adding (\ref{eq32}) and  (\ref{eq322}) together gives
  $$x^q+x-\delta-\delta^q=b+b^q.$$

  Let $c=b+\delta$. Then $\Tr(c)=\Tr(\delta)$ and
  \begin{equation}
  (x-c)^q=-(x-c).
  \end{equation}
  It follows that
  \begin{equation}\label{eqn-tty}
   x^{q^s}=(-1)^s(x-c) -(-c)^{q^s}, \ \ s=1,2 \ldots, 4k-1.
  \end{equation}
  With the help of these equations, we obtain
  \begin{eqnarray*}
  \lefteqn{g(x)^q-g(x)+ax-a(b+\delta)} \\
  & = &  g(x)^q-g(x)+a(x-c) \\
  & = & (x-c)(a-\Tr(\delta)) + \sum_{i=1}^k \left( (-c)^{q^{2k+4i-2}} - (-c)^{q^{2k+4i-4}} \right).
  \end{eqnarray*}
  Hence, (\ref{eq32}) has a unique solution if and only if $\Tr(\delta)\ne a$.
  This completes the proof.
\end{proof}

 Theorem \ref{thm-dy182} is a generalization of Theorem 6 in \cite{ZH}. Similarly, we have the
following theorem whose proof is similar to that of Theorem \ref{thm-dy182} and is omitted.

 \begin{theorem}
 Let $n=4k$  and $\delta$ be an element of $\mathbb{F}_{q^n}$.
 Let
 $$g(x)=x^{q+q^{2k+1}}+\cdots+x^{q^{2k-1}+q^{4k-1}}.$$
 Then the polynomial
 $$f(x)=g^q(x^q-x+\delta)+ax,$$
 where $0\ne a\in\mathbb{F}_{q}$,  permutes $\mathbb{F}_{q^n}$ if and only if  $\Tr(\delta) \ne -a$.
\end{theorem}

 We also have the following theorem.

 \begin{theorem} Let $h(x)$ be a polynomial over $\mathbb{F}_{q^6}$ and $L(x)$ be a {\em $q$-polynomial} over $\gf_{q}$.  Then the polynomial
 $$f(x)=h(x^{q^2}-x^q+x+\delta)^{q^4}+h(x^{q^2}-x^q+x+\delta)^{q^3}-h(x^{q^2}-x^q+x+\delta)^{q}
 -h(x^{q^2}-x^q+x+\delta)+L(x)$$permutes $\mathbb{F}_{q^6}$ if and only if $L(x)$ permutes
$\mathbb{F}_{q^6}$, where $\delta \in \mathbb{F}_{q^6}$.
\end{theorem}

\begin{proof}
It follows from Lemma \ref{lem-1.1} and the following commutative diagram:
$$
\xymatrix{
  \mathbb{F}_{q^6} \ar[d]_{x^{q^2}-x^q+x+\delta} \ar[r]^{f}
                & \mathbb{F}_{q^6} \ar[d]^{x^{q^2}-x^q+x}  \\
  S \ar[r]_{L(x)-L(\delta)}
                & \bar{S}      }
$$
The details of the proof are omitted.
\end{proof}

\begin{theorem}
Let $h(x)$ be a polynomial over $\mathbb{F}_{q^6}$ and $L(x)$ be a {\em $q$-polynomial} over $\gf_{q}$. Then the polynomial
 $$f(x)=h(x^{q^2}+x^q+x+\delta)^{q^4}-h(x^{q^2}+x^q+x+\delta)^{q^3}+h(x^{q^2}-x^q+x+\delta)^{q}
 -h(x^{q^2}-x^q+x+\delta)+L(x)$$permutes $\mathbb{F}_{q^6}$ if and only if $L(x)$ permutes
$\mathbb{F}_{q^6}$, where $\delta \in \mathbb{F}_{q^6}$.
\end{theorem}

\begin{proof}
It follows from Lemma \ref{lem-1.1} and the following commutative diagram
$$
\xymatrix{
  \mathbb{F}_{q^6} \ar[d]_{x^{q^2}+x^q+x+\delta} \ar[r]^{f}
                & \mathbb{F}_{q^6} \ar[d]^{x^{q^2}+x^q+x}  \\
  S \ar[r]_{L(x)-L(\delta)}
                & \bar{S}      }
$$
The details of the proof are omitted.
\end{proof}

At the end of this section, we present the following more generic theorem on permutation polynomials.

\begin{theorem}\label{thm-Gener}
Let $n$ be a positive integer. Let
$
L(x) \in\mathbb{F}_{q^n}[x]
$
be a {\em $q$-polynomial} over $\gf_{q^n}$ such that $\gcd(l(x), x^n-1)\ne1$, where $l(x)$ is the associated polynomial of $L(x)$. Let $a \in \mathbb{F}_{q^n}^*$ be a solution of the equation $L(x)=0$ and $h(x)$ be a polynomial with $h(x)^q=h(x)$. Let $L_1(x)\in\mathbb{F}_q[x]$ be a linearized polynomial.
 Then for every $\delta\in\mathbb{F}_{q^n}$, the polynomial
$$f(x)=g(L(x)+\delta)+L_1(x)$$
permutes $\mathbb{F}_{q^n}$ if and only if $L_1(x)$ permutes
$\mathbb{F}_{q^n}$.
\end{theorem}

\begin{proof}
To prove this theorem with Theorem \ref{thm-main3}, we define $A=\gf_{q^n}$ and
$$
\psi(x)=L(x)+\delta, \ \bar{\psi}(x)=L(x), \ S=\{\psi(x): x \in A\}, \  \bar{S}=\{\bar{\psi}(x): x \in A\}.
$$
We further define
$$
u(x)=L_1(x), \ v(x)=g(L(x)+\delta), \mbox{ and } h(x)=L_1(x) - L_1(\delta).
$$

With the assumptions in the theorem, we known that $g(x)=ah(x)$ is a polynomial such that $L(g(x))=0$. One
can verify that all the conditions in Theorem \ref{thm-main3} are satisfied. The desired conclusion then follows
from Theorem \ref{thm-main3}.
\end{proof}

Theorem \ref{thm-Gener} is a generalization of Theorem 5.6(a) in \cite{AGW}.

\section{Permutation polynomials of the form $(ax^{q^k}-bx+\delta)^{\frac{q^n+1}{2}}+ax^{q^k}+bx$}\label{sec-explicitpps}

In this section, we investigate permutation polynomials of the form $(ax^{q^k}-bx+\delta)^{\frac{q^n+1}{2}}+ax^{q^k}+bx$,
and generalize the permutation polynomials of the same form described in \cite{ZZH}, \cite{ZH}, and \cite{LHT}.

Let $\alpha$ be a primitive element of $\mathbb{F}_{q^n}$, where $q$ is a prime power, define $D_0=<\alpha^2>$, the multiplicative group generated by $\alpha^2$, and
$D_1=\alpha D_0$. Then $\mathbb{F}_{q^n}=\{0\}\cup D_0\cup D_1$.

\begin{theorem} Let $q$ be an odd prime power, $n, k$ be positive integers, $a, b,  \delta\in\mathbb{F}_{q^n}$, $ab\ne0$. Then  $(ax^{q^k}-bx+\delta)^{\frac{q^n+1}{2}}+ax^{q^k}+bx$ is a permutation polynomial over $\mathbb{F}_{q^n}$ if and only if $ab\in D_0$.
\end{theorem}
\begin{proof} For any given $u\in\mathbb{F}_{q^n}$, we consider  the following equation
\begin{equation}\label{eqn-40}(ax^{q^k}-bx+\delta)^{\frac{q^n+1}{2}}+ax^{q^k}+bx=u.\end{equation} Assume that $x$ is a solution to (\ref{eqn-40}), we distinguish among the following three cases.

{\bf Case 1:} $ax^{q^k}-bx+\delta=0$. By (\ref{eqn-40}), we have $ax^{q^k}+bx=u$. Then these two equations lead to $x=\frac{1}{2b}(u+\delta)$ and $x^{q^k}=\frac{1}{2a}(u-\delta)$ which imply
$$\frac{a}{b^{q^k}}(u+\delta)^{q^k}=u-\delta.$$

{\bf Case 2:} $ax^{q^k}-bx+\delta\in D_0$. In this case, (\ref{eqn-40}) is reduced to $ax^{q^k}-bx+\delta+ax^{q^k}+bx=u$, i.e., $x^{q^k}=\frac{1}{2a}(u-\delta)$. Then we have
$x=\frac{1}{2}\left(\frac{u-\delta}{a}\right)^{q^{n-k}}$ and
$$ax^{q^k}-bx+\delta=\frac{1}{2}(u+\delta)-\frac{b}{2}\left(\frac{u-\delta}{a}\right)^{q^{n-k}}.$$

{\bf Case 3:} $ax^{q^k}-bx+\delta\in D_1$. In this case, (\ref{eqn-40}) is reduced to $-(ax^{q^k}-bx+\delta)+ax^{q^k}+bx=u$, i.e., $x=\frac{1}{2b}(u+\delta)$. Then we have
$$ax^{q^k}-bx+\delta=\frac{a}{2}\frac{(u+\delta)^{q^k}}{b^{q^k}}-\frac{1}{2}(u-\delta).$$

If we denote $\Delta=\frac{1}{2}(u+\delta)-\frac{b}{2}\left(\frac{u-\delta}{a}\right)^{q^{n-k}}$ and $\Delta_1=\frac{a}{2}\frac{(u+\delta)^{q^k}}{b^{q^k}}-\frac{1}{2}(u-\delta)$, then
$$\Delta^{q^k}a=\Delta_1 b^{q^k}.$$

First, if Case 1 occurs, i.e., $\Delta=\Delta_1=0$, then both Case 2 and Case 3 cannot happen.

If $ab\in D_0$ and $u$ is an element such that $\Delta\in D_0$, then (\ref{eqn-40}) has a solution $x=\frac{1}{2}\left(\frac{u-\delta}{a}\right)^{q^{n-k}}$. If $ab\in D_0$ and $u$ is an element such that $\Delta\in D_1$, then $\Delta_1\in D_1$ and  (\ref{eqn-40}) has a solution $x=\frac{1}{2b}(u+\delta)$. This implies that $(ax^{q^k}-bx+\delta)^{\frac{q^n+1}{2}}+ax^{q^k}+bx$ is a permutation polynomial over $\mathbb{F}_{q^n}$

 If $ab\in D_1$ and $u$ is an element such that $\Delta\in D_0$, then $\Delta_1\in D_1$, and so (\ref{eqn-40}) has two solutions $x=\frac{1}{2}\left(\frac{u-\delta}{a}\right)^{q^{n-k}}$ and $x=\frac{1}{2b}(u+\delta)$. If $ab\in D_1$ and $u$ is an element such that $\Delta\in D_1$, then $\Delta_1\in D_0$ and  (\ref{eqn-40}) has no solutions. This completes the proof.\end{proof}

\section{Summary and concluding remarks}

Recently, it has been a hot topic to construct permutation polynomials over finite fields of specific forms \cite{CH,CK08,CK09,DXY,Ky10, LHT,Hou,HMSY,Wa07,YD,Ma10}. The main contributions
of this paper are the general theorems on permutation polynomials described in Section \ref{sec-mainthms}
and explicit permutation polynomials documented in Sections \ref{sec-mainthms}  and
\ref{sec-explicitpps}. Many of the results presented in this paper are extensions and generalizations
of earlier results on permutation polynomials in the references of this paper.

To employ the theorems in this paper for the construction of more permutation polynomials, we need
to construct linearized permutation polynomials $L(x)$. Let $l(x)$ be any polynomial of degree at most
$n-1$ over $\gf_q$ with $\gcd(l(x), x^n-1)$, and let $L(x)$ denote its $q$-associate. It then follows
from Lemma \ref{lem-gcdaa} that $L(x)$ is
a linearized PP over $\gf_{q^n}$. The reader is referred to \cite{YD} for further information on this method.

\end{document}